\documentclass[11pt,letterpaper]{article}
\usepackage[top=1in, bottom=1in, left=1in, right=1in]{geometry}

\usepackage[utf8]{inputenc} 
\usepackage[T1]{fontenc}    
\usepackage{hyperref}       
\usepackage{url}            
\usepackage{booktabs}       
\usepackage{amsfonts}       
\usepackage{nicefrac}       
\usepackage{microtype}      
\usepackage{graphicx}
\usepackage{enumitem}
\usepackage[square,sort,comma,numbers]{natbib}
\usepackage{algorithm, algorithmic}

\title{Exact Inference with Latent Variables \\ in an Arbitrary Domain}

\author{
  \textbf{Chuyang Ke}\\Department of Computer Science\\Purdue University\\\texttt{cke@purdue.edu}
  \and 
  \textbf{Jean Honorio}\\Department of Computer Science\\Purdue University\\\texttt{jhonorio@purdue.edu}
}

\usepackage{amsmath, amsthm, amssymb}
\usepackage{bm}
\usepackage{color}
\usepackage{ulem}

\let\emph\textit

\def\R{{\mathbb{R}}}

\def\1{\textbf{1}}

\newcommand{\econst}{\mathrm{e}}

\newcommand{\onevct}{\vct{1}}

\newcommand{\onemtx}{\onevct_n \onevct_n^\top}
\newcommand{\zeromtx}{0}
\newcommand{\Imtx}{\mathbf{I}}

\newcommand{\onemtxm}{\onevct_{m} \onevct_{m}^\top}

\newtheorem{theorem}{Theorem}

\newtheorem{corollary}{Corollary}
\newtheorem{lemma}{Lemma}
\newtheorem{definition}{Definition}

\newtheorem{claim}{Claim}

\newcommand{\Prob}[2][]{\mathbb{P}_{#1}\left\{ {#2} \right\}}
\newcommand{\Expect}[2][]{\mathbb{E}_{#1}\left[ #2 \right]}
\newcommand{\Var}[2][]{\operatorname{Var}_{#1}\left[ #2 \right]}

\newcommand{\abs}[1]{\left\vert {#1} \right\vert}

\newcommand{\norm}[1]{\left\Vert {#1} \right\Vert}
\newcommand{\normsq}[1]{{\norm{#1}}^2}
\newcommand{\normf}[1]{{\norm{#1}}_F}

\newcommand{\diag}[1]{\operatorname{diag}{#1}}
\newcommand{\tr}[1]{\operatorname{tr}{#1}}
\newcommand{\rank}[1]{\operatorname{rank}{#1}}
\newcommand{\kl}[1]{D_{\text{KL}}{#1}}

\newcommand{\st}{\operatorname*{subject\; to}}
\newcommand{\maximize}{\operatorname*{maximize}}
\newcommand{\minimize}{\operatorname*{minimize}}

\newcommand{\vct}[1]{{#1}}
\newcommand{\mtx}[1]{{#1}}

\newcommand{\lmax}{\lambda_{\max}}

\newcommand{\lsec}{\lambda_{2}}

\newcommand{\Snp}{{\mathcal{S}_{+}^n}}
\newcommand{\Rnp}{{\R_{+}^n}}
\newcommand{\Rmp}{{\R_{+}^m}}
\newcommand{\sphere}{{\mathbb{S}^{n-1}}}

\newcommand{\my}{\mtx{Y}}
\newcommand{\mya}{\mtx{Y}^\ast}
\newcommand{\mw}{\mtx{W}}
\newcommand{\ma}{\mtx{W}}
\newcommand{\mx}{\mtx{X}}
\newcommand{\mall}{\mtx{WX}}

\newcommand{\inprod}[2][]{\left\langle {#1},{#2} \right\rangle}
\newcommand{\spa}[1]{\operatorname*{span}\left({#1}\right)}

\date{}
\begin{document}
\maketitle

\begin{abstract}
We analyze the necessary and sufficient conditions for exact inference of a latent model. In latent models, each entity is associated with a latent variable following some probability distribution. The challenging question we try to solve is: can we perform exact inference without observing the latent variables, even without knowing what the domain of the latent variables is? We show that exact inference can be achieved using a semidefinite programming (SDP) approach without knowing either the latent variables or their domain. Our analysis predicts the experimental correctness of SDP with high accuracy, showing the suitability of our focus on the Karush-Kuhn-Tucker (KKT) conditions and the spectrum of a properly defined matrix. As a byproduct of our analysis, we also provide concentration inequalities with dependence on latent variables, both for bounded moment generating functions as well as for the spectra of matrices. To the best of our knowledge, these results are novel and could be useful for many other problems. 
\end{abstract}

\allowdisplaybreaks


\section{Introduction} 
Generative network models have become a powerful tool for researchers in various fields, including data mining, social sciences, and biology \citep{goldenberg2010survey,fortunato2010community}. With the emergence of social media in the past decade, researchers are now exposed to millions of records of interaction generated on the Internet everyday. One can note that the generic structure and organization of social media resemble certain network models, for instance, the Erdos-Renyi model, the stochastic block model, the latent space model, the random dot product model \citep{goldenberg2010survey, newman2002random, young2007random}. The analogy comes from the fact that, in a social network each user can be modeled as an \emph{entity}, and the interaction of users can be modeled as edges. One common assumption is that nodes belong to different groups. In social networks this can be users' political view, music genre preferences, or whether the user is a cat or dog person. Another common assumption, often referred to as \emph{homophily} in prior literature, suggests that entities from the same group are more likely to be connected with each other than those from different groups \citep{goldenberg2010survey,hoff2008modeling, krivitsky2009representing}. The core task of \emph{inference}, also known as graph partitioning, is to partition the nodes into groups based on the observed interaction information \citep{abbe2017community, ke2018information, fortunato2010community}.  

In this paper, we are particularly interested in the class of \emph{latent models} beyond graphs, with latent variables in arbitrary domains.
In a latent model, every entity belongs to one of $k$ groups. 
Every entity is associated with a latent variable in some arbitrary latent domain. It is natural to assume that for entities from the same group, their associated latent variables follow the same probability distribution. The latent model is equipped with a function to measure the homophily of two latent variables. Finally, two entities have some affinity score depending on their homophily in the latent domain. In other words, similar entities are more likely to have a higher affinity score. We want to highlight that, for the particular case of binary (i.e., $\{0,1\}$) affinity scores, the latent model is a random graph model. The challenging problem problem we try to solve is to infer the true group assignments \emph{without} observing the the latent variables nor knowing the latent domain.

In the past decade there have actually existed a large amount of literature on network models, and most focus on the class of \emph{fully observed models}, for example, the Erdos-Renyi Model, and the Stochastic Block Model. These models are called ``fully observed'', because there are no latent variables, and edges are generated based on the agreement of entity labels. 
Some efficient algorithms have also been proposed for inference in these fully observed models \citep{abbe2016exact, bandeira2018random, hajek2016achieving, chen2014statistical}.  
On the other hand, there is limited research on the class of latent models. Researchers have motivated various network models with latent variables, including the latent space model \citep{hoff2002latent}, the exchangeable graph model \citep{goldenberg2010survey}, the dot product model \citep{nickel2008random}, the uniform dot product model \citep{young2007random}, and the extremal vertices model \citep{daudin2010model}. However to the best of our knowledge, no efficient polynomial time algorithms with formal guarantees have been proposed or analyzed for \emph{exact inference} in latent models. 

In this paper we address the problem of exact inference in latent models with arbitrary domains. More specifically, our goal is to correctly infer the group assignment of every entity in a latent model without observing the latent variables or the latent domain. We also propose a polynomial-time algorithm for exact inference in latent models using semidefinite programming (SDP). 
We want to highlight that many techniques used in the analysis of fully observed models do not directly apply to latent models. This is because in latent models, affinities are no longer statistically independent. As a result, latent models are more challenging to analyze than fully observed models, such as the stochastic block model.

While SDP has been heavily proposed for different machine learning problems, our goal in this paper is to study the optimality of SDP for our more challenging model. Our analysis focuses on Karush-Kuhn-Tucker (KKT) conditions and the spectrum of a carefully constructed primal-dual certificate. For convex problems including SDPs, the KKT conditions are sufficient and necessary for strong duality and optimality \citep{boyd2004convex}.
To the best of our knowledge, we are providing the first polynomial time method for a generally computationally hard problem with formal guarantees. In general, problems involving latent variables are computationally hard and nonconvex, for instance, learning restricted Boltzmann machines \citep{long2010restricted} or structural Support Vector Machines with latent variables \citep{yu2009learning}. 
It is worth mentioning that theoretical computer science typically assumes arbitrary inputs ("worst-case" computationally hard), whereas we assume inputs are generated by a probabilistic generative model. Our results could be seen as "average-case" polynomial time: we provide exact inference conditions with respect to the model parameters $(p,q)$.

\textbf{Summary of our contributions.} We provide a series of novel results in this paper:
\begin{itemize}
    \item We propose the definition of the latent model class, which is highly general and subsumes several latent models from prior literature (see Table \ref{tab:my-table}).
    \item We provide the first polynomial time algorithm for a generally computationally hard problem with formal guarantees. We also analyze the sufficient conditions for exact inference in latent models using a semidefinite programming approach. 
    \item For completeness, we provide an information-theoretic lower bound on exact inference, and we analyze when nonconvex maximum likelihood estimation is correct.
    \item As a byproduct of our analysis, we provide concentration inequalities with dependence on latent variables, both for bounded moment generating functions as well as for the spectra of matrices. To the best of our knowledge, these results are novel and could be useful for many other problems.
\end{itemize}
\begin{table*}[t]
\centering
\begin{tabular}{|c|c|c|}
\hline
Models                        & $\mathcal{X}$                                                           & $f(x_i,x_j)$ \\ \hline
Latent space model \citep{hoff2002latent}           & $\R^d$                                                                  & $\exp(-\normsq{x_i-x_j})$ \\ \hline
Exchangeable graph model \citep{goldenberg2010survey}           & $\{0,1\}^d$                                                             & $\exp(-\norm{x_i-x_j}_1)$ \\ \hline
Dot product graph (DPG) \citep{nickel2008random}           & $\R^d$                                                                  & $g(x_i \cdot x_j)$ \\ \hline
Uniform DPG \citep{young2007random}   & $[0,1]^d$                                                               & $g(x_i \cdot x_j)$ \\ \hline
Extremal vertices model \citep{daudin2010model}                          & $\{ x \mid x \in \R^d,  x_i \geq 0, \sum_{i=1}^d x_i = 1 \}$ & $g(x_i \cdot x_j)$ \\ \hline
Kernel latent variable model & $\R^d$, sets, graphs, text, etc.                                                                     & $ g(K(x_i,x_j))$ \\ \hline
\end{tabular}
\caption{Comparison of various latent models. In dot product models, $g:\R\to [0,1]$ is a function that normalizes dot products to the range of $[0,1]$. In kernel models, $K : \mathcal{X} \times \mathcal{X} \to \R$ is an arbitrary kernel function.}
\label{tab:my-table}
\end{table*}

\section{Preliminaries}
In this section, we introduce the notations that will be used in later sections. First we provide the definition of the class of latent models.

\begin{definition}[Class of latent models]
A model $\mathcal{M}$ is called a \emph{latent model} with $n$ entities and $k$ clusters, if $\mathcal{M}$ is equipped with structure $(\mathcal{X}, f, \mathcal{P})$ satisfying the following properties:
\begin{itemize}
    \setlength\itemsep{0em}
    \item $\mathcal{X}$ is an arbitrary latent domain;
    \item $f: \mathcal{X} \times \mathcal{X} \to [0,1]$ is a homophily function, such that $f(\vct{x}, \vct{x}') = f(\vct{x}', \vct{x})$;
    \item $\mathcal{P} = (\mathcal{P}_1, \dots, \mathcal{P}_k)$ is the collection of $k$ distributions with support on $\mathcal{X}$.
\end{itemize}
For simplicity we consider the balance case in this paper: each cluster has the same size $m := n/k$. Let $Z^\ast \in \{0,1\}^{n \times k}$ be the true cluster assignment matrix, such that $Z_{ij}^\ast = 1$ if entity $i$ is in cluster $j$, and $Z_{ij}^\ast = 0$ otherwise. For every entity $i$ in cluster $j$, nature randomly generates a latent vector $x_i \in \mathcal{X}$ from distribution $P_j$. 
A random observed affinity matrix $W \in [0,1]^{n \times n}$ is generated, such that the conditional expectation fulfills $\Expect{W_{ij} | x_i, x_j} = f(x_i,x_j)$.
\label{def:class}
\end{definition}

\textbf{Remark.}
We use $[0,1]$ for $f$ and $W$ for clarity of exposition. Our results can be trivially extended to a general domain $[0,B]$ for $B>0$ using the same techniques in later sections.

\textbf{Remark.} 
A particular case of the latent model is a random graph model, in which every entry $W_{ij}$ in the affinity matrix is binary (i.e., $W_{ij} \in \{0,1\}$) and generated from a Bernoulli distribution with parameter $f(x_i, x_j)$.

\textbf{Our definition of latent models is highly general.} In Table \ref{tab:my-table}, we illustrate several latent models motivated from prior literature that can be subsumed under our model class by properly defining $\mathcal{X}$ and $f$. 

\textbf{In latent models, affinities are not independent if not conditioning on the latent variables.} For example, suppose $i, j$ and $k$ are three entities. In fully observed models the affinities $W_{ij}$ and $W_{ik}$ are independent, but this is not true in latent models, as shown graphically in Figure \ref{fig:models}. This motivates our following definition of \emph{latent conditional independence (LCI)}.

\begin{definition}[Latent Conditional Independence]
We say random variables $Y = (y_1,\dots, y_n)$ are \emph{latently conditional independent} given $X$, if $y_1, \dots, y_n$ are conditional independent given the unobserved latent random variable $X$.
\end{definition}

\begin{figure*}[t]
    \centering
    \includegraphics[width=0.9\textwidth]{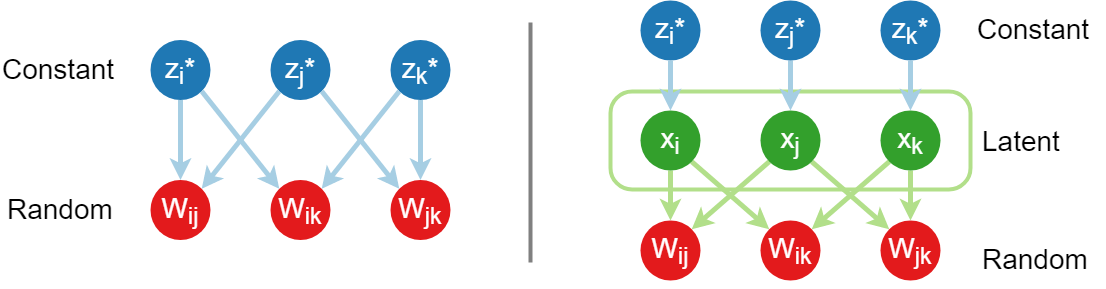}
    \caption{Comparison of fully observed models (left) and latent models (right). The blue nodes are the true (constant) labels $Z^\ast$. The green nodes are the unobserved latent random variables $\mx$ in our model. The red nodes are entries of the observed matrix $\ma$. In latent models, affinities are not independent without conditioning on latent variables. We say $W$'s are latently conditional independent given $X$.}
    \label{fig:models}
\end{figure*}

\subsection{Notations}

We denote $[n] := \{1,2,\dots, n\}$. We use $\Snp$ to denote the $n$-dimensional positive semidefinite matrix cone, and $\Rnp$ to denote the $n$-dimensional nonnegative orthant.

For simplicity of analysis, we use $z_i \in \{0,1\}^k$ to denote the $i$-th row of $Z$, and $z^{(i)}\in \{0,1\}^n$ to denote the $i$-th column of $Z$. We use $X = (x_1, \dots, x_n)$ to denote the collection of latent variables.

Regarding eigenvalues of matrices, we use $\lambda_i(\cdot)$ to refer to the $i$-th smallest eigenvalue, and $\lmax(\cdot)$ to refer to the maximum eigenvalue.

Regarding probabilities $\Prob[\ma]{\cdot}, \Prob[\mx]{\cdot}$, and $\Prob[\mall]{\cdot}$, the subscripts indicate the random variables. Regarding expectations $\Expect[\ma]{\cdot}, \Expect[\mx]{\cdot}$, and $\Expect[\mall]{\cdot}$, the subscripts indicate which variables we are averaging over.
We use $\Prob[\ma]{\cdot \mid \mx}$ to denote the conditional probability with respect to $\ma$ given $\mx$, and $\Expect[\ma]{\cdot \mid \mx}$ to denote the conditional expectation with respect to $\ma$ given $\mx$.

For matrices, we use $\norm{\cdot}$ to denote the spectral norm of a matrix, and $\normf{\cdot}$ to denote the Frobenius norm. We use $\tr(\cdot)$ to denote the trace of a matrix, and $\rank(\cdot)$ to denote the rank. We use the notation $\diag(a_1,\dots,a_n)$ to denote a diagonal matrix with diagonal entries $a_1,\dots,a_n$. We also use $\Imtx$ to refer to the identity matrix, and $\onevct_n$ to refer to an all-one vector of length $n$. We use $\sphere$ to denote the unit $(n-1)$-sphere.

Let $S_i \in [n]^m$ denote the index set of the $i$-th cluster. For any vector $v \in \R^n$, we define $v_{S_i}$ to be the subvector of $v$ on indices $S_i$. Similarly for any matrix $X \in \R^{n\times n}$, we define $X_{S_i S_j}$ to be the submatrix of $X$ on indices $S_i \times S_j$. Denote the shorthand notation $X_{S _i} := X_{S _i S_i}$.

Define $d_i(S_l) := \sum_{j\in S_l} W_{ij}$ to be the degree of entity $i$ with respect to cluster $l$. Define shorthand notation $d_i$ to be the degree of entity $i$ with respect to its own cluster. Algebraically, we have $d_i := \sum_j W_{ij} z_i^{\ast\top} z_j^\ast$. We also denote $D := \diag(d_1, \dots, d_n)$.

In the following sections we will frequently use the expected values related to the observed affinity matrix $W$. It would be tedious to derive every expression from $(\mathcal{X},f,\mathcal{P})$. To simplify this, we introduce the following induced model parameters, which will be used throughout the paper.
\begin{definition}[Induced model parameters]
In a latent model $\mathcal{M}$ equipped with structure $(\mathcal{X}, f, \mathcal{P})$, one can derive the induced parameters $(p, q)$ defined as
\begin{align*}
p := \Expect[\mtx{X}]{f(\vct{x}_i, \vct{x}_j) \mid z^\ast_i = z^\ast_j} \,, \qquad
q := \Expect[\mtx{X}]{f(\vct{x}_i, \vct{x}_j) \mid z^\ast_i \neq z^\ast_j} \,.
\end{align*}
Note that both $p,q \in [0,1]$.
\label{def:induced}
\end{definition}


\subsection{LCI Concentration Inequalities}
Here we provide new concentration inequalities with dependence on latent variables, both for bounded moment generating functions as well as for the spectra of matrices.

\begin{lemma}[LCI Tail Bound]
Consider a finite sequence of random variables $\{x_i\}$ that are LCI given $Y$. Assume that $\Expect[x_i Y]{x_i} = \mu_i$, and $\Expect[x_i]{\econst^{t (x_i - \mu_i)} \mid Y} \leq \econst^{t^2 \sigma_i^2 / 2}$ for all $Y$.
Then for all positive $\epsilon$,
\[
\Prob[XY]{\sum_i (x_i - \mu_i) \geq \epsilon}
\leq 
\exp\left( -\frac{\epsilon^2}{2\sum_i \sigma_i^2} \right) \,.
\]
\label{thm:LCI_tail}
\end{lemma}

\begin{corollary}[LCI Hoeffding's Inequality]
Consider a finite sequence of random variables $\{x_i\}$ that are LCI given $Y$. Assume that  and $x_i \in [a_i,b_i]$ almost surely, and $\Expect[x_i Y]{x_i} = \mu_i$.
Then for all positive $\epsilon$,
\begin{align*}
\Prob[XY]{\sum_i (x_i - \mu_i) \geq \epsilon}   
\leq& \exp\left( -\frac{2\epsilon^2}{\sum_i (b_i-a_i)^2} \right) \,.
\end{align*}
\label{cor:LCI_hoeffding}
\end{corollary}

\begin{corollary}[LCI Bernstein Inequality]
Consider a finite sequence of random variables $\{x_i\}$ that are LCI given $Y$. Assume that $\abs{x_i} \leq M$ almost surely, $\Expect[x_i]{x_i \mid Y} = 0$, and $\Var[x_i]{x_i \mid Y} \leq \nu_i^2$  for all $Y$.
Then for all positive $\epsilon$,
\[
\Prob[XY]{\sum_i x_i \geq \epsilon}
\leq 
\exp\left( -\frac{\epsilon^2 / 2}{\sum_i \nu_i^2 + M\epsilon/3} \right) \,.
\]
\label{cor:LCI_bern}
\end{corollary}

\begin{lemma}[LCI Matrix Tail Bound]
Consider a finite sequence of random symmetric matrices $\{X_i\}$ with dimension $d$ that are LCI given $Y$. Assume there is a function $g:(0,\infty) \to [0, \infty]$ and a sequence $\{A_i\}$ of fixed symmetric matrices that satisfy the relations
$
\Expect[X_i]{\econst^{\theta X_i} \mid Y} \preceq \econst^{g(\theta) \cdot A_i}
$
for $\theta > 0$ and for all $Y$. Define the scale parameter 
$
\rho := \lmax\left(\sum_i A_i\right)
$.
Then for all positive $\epsilon$,
\begin{equation*}
\Prob[XY]{\lmax\left(\sum_i X_i\right) \geq \epsilon }
\leq d \cdot \inf_{\theta > 0} \econst^{-\theta \epsilon + g(\theta) \cdot \rho} \,.
\end{equation*}
\label{corollary_3.7}
\end{lemma}

\begin{corollary}[LCI Matrix Bernstein Inequality]
Consider a finite sequence of random symmetric matrices $\{X_i\}$ with dimension $d$ that are LCI given $Y$. Assume that $\Expect[X_i]{X_i \mid Y} = \zeromtx$ for all $Y$, and $\lmax(X_i) \leq R$ almost surely. 
Also assume that the norm of the total variance $\norm{\sum_i \Expect[X_i]{X_i^2\mid Y}} \leq \sigma^2$ for all $Y$. 
Then Then for all positive $\epsilon$,
\begin{equation*}
\Prob[XY]{\lmax\left(\sum_i X_i\right) \geq \epsilon}
\leq d \cdot \exp\left( -\frac{\epsilon^2/2}{\sigma^2+R\epsilon/3} \right) \,.
\end{equation*}
\label{thm:LCI_mtxBern}
\end{corollary}
\section{Polynomial-Time Regime with Semidefinite Programming}
In this section we investigate the sufficient conditions for exactly inferring the group assignment of entities in latent models. An algorithm achieves \emph{exact inference} if the recovered group assignment matrix $Z \in \{0,1\}^{n\times k}$ is identical to the true assignment matrix $Z^\ast$ up to permutation of its columns (without prior knowledge it is impossible to infer the order of groups). 

\textbf{Overview of the proof.} 
Our proof starts by looking at a \emph{maximum likelihood estimation (MLE)} problem \eqref{opt:mle2}, which cannot be solved efficiently (for more details see Section 4). We relax the MLE problem \eqref{opt:mle2} to problem \eqref{opt:mleY} (matrix-form relaxation), then to problem \eqref{opt:primal} (convex SDP relaxation). We ask under what conditions the relaxation holds (i.e., returns the groundtruth).
Our analysis proves that, if the statistical conditions in Theorem \ref{sdp:theorem} are satisfied, by solving the relaxed convex optimization problem \eqref{opt:primal}, one can recover the true group assignment $Z^\ast$ perfectly and efficiently with probability tending to $1$. 

Our analysis can be broken down into two parts. In the first part we demonstrate that the exact inference problem in latent models can be relaxed to a semidefinite programming problem. It is well-known that SDP problems can be solved efficiently \citep{boyd2004convex}. Motivated by \citep{amini2018semidefinite} we employ \emph{Karush-Kuhn-Tucker (KKT) conditions} in our proof to construct a pair of primal-dual certificates, which shows that the SDP relaxation leads to the optimal solution under certain deterministic spectrum conditions. In the second part we analyze the statistical conditions for exact inference to succeed with high probability.

\subsection{SDP Relaxation}
We first consider a maximum likelihood estimation approach to recover the true assignment $Z^\ast$. 
The use of MLE in graph partitioning and community detection literature is customary \citep{bandeira2018random, abbe2016exact, chen2014statistical}. The motivation is to find cluster assignments, such that the number of edges within clusters is maximized.
Recall that $z_i \in \{0,1\}^k$ is the $i$-th row of $Z$, and $z^{(i)}\in \{0,1\}^n$ is the $i$-th column of $Z$. Given the observed matrix $W$, the goal is to find a binary assignment matrix $Z$, such that
$\sum_{i,j} W_{ij} z_i^\top z_j$
is maximized. In the matrix form, MLE can be cast as the following optimization problem:
\begin{align}
\maximize_{Z} \qquad & \langle W, ZZ^\top\rangle  \nonumber\\
\st \qquad &Z \in \{0,1\}^{n\times k} , \quad
 Z^\top \onevct_n = m  \onevct_k , \quad
Z\onevct_k = \onevct_n \,.
\label{opt:mle2}
\end{align}
where the last two constraints enforce that each entity is in one of the $k$ groups, and each group has size $m = n/k$. 

Problem \eqref{opt:mle2} is nonconvex and hard to solve because of the $\{0,1\}$ constraint. In fact, in the case of two clusters ($k=2$) and $0$-$1$ weights, the MLE formulation reduces to the Minimum Bisection problem, which is known to be NP-hard \citep{GAREY1976237}. 
To relax it, we introduce the cluster matrix $Y = ZZ^\top$. One can see that $Y$ is a rank-$k$, $\{0,1\}$ positive semidefinite matrix. Each entry is $1$ if and only if the corresponding two entities are in the same group ($z_i = z_j$). Similarly we can define $Y^\ast = Z^\ast Z^{\ast\top}$ for the true cluster matrix. Then the optimization problem becomes
\begin{align}
\maximize_{Y} \qquad & \langle W, Y\rangle  \nonumber\\
\st \qquad 
& Y_{ii} = 1\,, \forall i \in [n] ,\quad
 Y \onevct_n = m  \onevct_n ,\quad
 \onevct_n^\top Y = m  \onevct_n^\top \nonumber \\
& Y \succeq_\Snp 0 ,\quad
 Y \succeq_\Rnp 0 ,\quad
 \rank(Y) = k \,.
\label{opt:mleY}
\end{align}
Problem \eqref{opt:mleY} is still nonconvex because of the rank constraint. By dropping this constraint, we obtain the main SDP problem:
\begin{align}
\maximize_{Y} \qquad & \langle W, Y\rangle  \nonumber\\
\st \qquad 
& Y_{ii} = 1, \forall i \in [n] ,\quad
 Y \onevct_n = m  \onevct_n ,\quad
\onevct_n^\top Y = m  \onevct_n^\top \nonumber \\
& Y \succeq_\Snp 0 ,\quad
 Y \succeq_\Rnp 0 \,.
\label{opt:primal}
\end{align}

Problem \eqref{opt:primal} is now convex and can be solved efficiently. A natural question is: under what circumstances the optimal solution to \eqref{opt:primal} will match the solution to the original problem \eqref{opt:mle2}? To answer the question, we take a primal-dual approach. One can easily see there exists a strictly feasible $Y$ for the constraints in \eqref{opt:primal}. Thus Slater's condition guarantees \emph{strong duality} \citep{boyd2004convex}. We now proceed to derive the dual problem. 

\begin{lemma}[Lagrangian Dual]
The dual problem of \eqref{opt:primal} is
\begin{align}
\minimize_{v,A,\Gamma} \qquad & \tr(A) + 2m v^\top \onevct_n  \nonumber\\
\st \qquad 
& A-W + v\onevct_n^\top + \onevct_n v^\top - \Gamma \succeq_\Snp 0 \nonumber\\
& A \text{ is diagonal},\quad
 \Gamma_{S_i} = 0 \,, \forall i \in [k] ,\quad
 \Gamma \succeq_\Rnp 0 \,.
\label{opt:dual}
\end{align} 
\label{lemma:dual}
\end{lemma}


We now construct the primal-dual certificates to close the duality gap between problem \eqref{opt:primal} and \eqref{opt:dual}.
\begin{lemma}[Primal-dual Certificates]
Let $P := \Imtx - \frac{1}{m}\onemtxm$ to be the projection onto the orthogonal complement of $\spa{\onevct_m}$. 
By setting the dual variables as follows
\begin{align*}
&v = \frac{\phi}{2} \onevct_n ,\quad
A = D - m\phi \Imtx \,, \\
&\Gamma_{S_i} = 0, \forall i\in [k] ,\quad
\Gamma_{S_i S_j} = \phi\onemtxm + PW_{S_i S_j}P - W_{S_i S_j}, \forall i \neq j \,, 
\end{align*}
where $\phi \in \R$ is a constant to be determined later, the duality gap between \eqref{opt:primal} and \eqref{opt:dual} is closed.
\label{lemma:certificates}
\end{lemma}

It remains to verify feasibility of the dual constraints in \eqref{opt:dual}. It is trivial to verify that $A = D - m\phi \Imtx$ is diagonal, and $\Gamma_{S_i} \succeq_\Rmp 0$. We now summarize the dual feasibility conditions.
\begin{lemma}[Dual Feasibility]
Let $v, A, \Gamma$ be defined as in Lemma \ref{lemma:certificates}. If 
\begin{equation}
\Lambda := D-m\phi \Imtx - W + \phi \onemtx - \Gamma \succeq_\Snp 0 \,,
\end{equation}
and 
\begin{equation}
\Gamma_{S_i S_j} \succeq_\Rmp 0
\end{equation}
for every $i,j \in [k]$ with $i\neq j$, then the dual constraints in \eqref{opt:dual} are satisfied.
\label{lemma:df}
\end{lemma}

\textbf{We also require the optimal solution to be unique.} This means $Y^\ast = Z^\ast Z^{\ast\top}$ should be the only optimal solution to problem \eqref{opt:primal}. To do so we look into the eigenvalues of $\Lambda$ defined in Lemma \ref{lemma:df}. It is easy to verify that every $z^{\ast(i)}$ is an eigenvector of $\Lambda$ with $\Lambda z^{\ast(i)} = 0$.
To ensure uniqueness, it is sufficient to require that all other $n-k$ eigenvalues of $\Lambda$ are strictly positive. 
We now provide the following lemma about uniqueness.

\begin{lemma}[Uniqueness]
The convex relaxed problem \eqref{opt:primal} achieves exact inference and outputs the unique optimal solution $Y = Y^\ast = Z^\ast Z^{\ast\top}$, if 
\begin{equation}
\lambda_{k+1} (\Lambda) > 0 \,.
\label{eq:eigen_condition}
\end{equation}
\label{lemma:uniqueness}
\end{lemma}

\textbf{Remark.} Why is the requirement of uniqueness reasonable? Because our latent models are generative, i.e., the ground truth $Z^\ast$ is unique and generates everything, including the latent variables $X$ and the observed matrix $W$ (see Figure \ref{fig:models}). From the perspective of optimization, in some cases there may exist multiple optimal solutions, but we are only interested in the cases in which the preexisting groundtruth $Z^\ast$ is returned. In fact, the requirement of uniqueness is customary in generative models \citep{abbe2016exact,bandeira2018random,chen2014statistical}.


Combining the results above, we now give the sufficient conditions for exact inference.

\begin{lemma}[Deterministic Sufficient Conditions]
Let $v, A, \Gamma$ be defined as in Lemma \ref{lemma:certificates}. If 
\begin{align}
& \Gamma_{S_i S_j} \succeq_\Rmp 0 \label{eq:entrygeq0}
\end{align}
for every $i,j\in [k]$ with $i \neq j$, and
\begin{align}
&\lambda_{k+1} (\Lambda) > 0 \,, \label{eq:eigenk+1}
\end{align}
then $Y^\ast = Z^\ast Z^{\ast\top}$ is the unique primal optimal solution to \eqref{opt:primal}, and $(v, A, \Gamma)$ is the dual optimal solution to \eqref{opt:dual}.
\label{theorem:sdp_deterministic}
\end{lemma}

Note that Lemma \ref{theorem:sdp_deterministic} gives the deterministic condition for our SDP relaxation to succeed. In the following two sections, we characterize the statistical conditions for \eqref{eq:entrygeq0} and \eqref{eq:eigenk+1} to hold with probability tending to $1$.


\subsection{Entrywise Nonnegativity of $\Gamma$}
In this section we analyze the statistical conditions for \eqref{eq:entrygeq0} to hold with high probability. From Lemma \ref{lemma:certificates} it follows that 
$\Gamma_{S_i S_j} = \phi\onemtxm + PW_{S_i S_j}P - W_{S_i S_j}, \forall i \neq j$. To ensure dual feasibility, it is necessary to ensure that every entry in $\Gamma_{S_i S_j}$ is nonnegative with high probability by setting a proper $\phi$.

We now present the condition for \eqref{eq:entrygeq0} to hold with high probability.
\begin{lemma}[Choice of $\phi$]
If $\phi \geq q + O\left(\sqrt{\frac{k\log n}{n}}\right)$, then $\Gamma_{S_i S_j} \succeq_\Rmp 0$ holds for every $i,j\in [k]$ with probability at least $1 - O\left(\frac{1}{n}\right)$.
\label{lemma:choice_phi}
\end{lemma}

\textbf{Remark.} To ensure nonnegativity, one may think about setting $\phi$ to be some sufficiently large constant (for example, set $\phi = 2$). This is not going to work, however, as the choice of $\phi$ also plays a critical role in the analysis of \eqref{eq:eigenk+1} in the next section. In order to obtain a tighter final result, it is necessary to pick the smallest possible $\phi$, without breaking the nonnegativity of $\Gamma$. For further details see Lemma \ref{lemma:sdp_expectation}.


\subsection{Statistical Conditions of Efficient Inference}
In this section we analyze the statistical conditions for \eqref{eq:eigenk+1} to hold with high probability. 
To do so, we first look at the expectation of $\Lambda$.

\begin{lemma}
It follows that
\begin{equation}
\lambda_{k+1} \left(\Expect[WX]{\Lambda}\right) = m (p-\phi) \,.    
\end{equation}
\label{lemma:sdp_expectation}
\end{lemma}

\textbf{Remark.} The expectation above shows why the choice of $\phi$ matters. With a larger $\phi$, one has less degree of freedom to work with, in terms of the concentration inequalities.

The next step is to show that the eigenvalue of $\Lambda$ will not deviate too much from its expectation, so that $\lambda_{k+1}(\Lambda)$ is greater than $0$ with high probability. In fact we have the following lemma.

\begin{lemma}
Assuming that $\phi < p$. To prove \eqref{eq:eigenk+1} holds with high probability, it is sufficient to prove
\begin{align}
\min_i (d_i - \Expect[WX]{d_i}) + \frac{m}{2} (p-\phi) &> 0 
\label{eq:d_con}
\end{align}
and
\begin{align}
-\lmax(W - \Expect[WX]{W}) + \frac{m}{2} (p-\phi) &> 0
\label{eq:w_con}
\end{align}
hold with high probability.
\label{lemma:dw}
\end{lemma}


We now present the statistical conditions for exact inference of latent models using semidefinite programming.
\begin{theorem}
In a latent model of $k$ clusters and $n$ entities, and with induced parameters $(p, q)$ as in Definition \ref{def:induced}, if 
\[
\frac{(p-q)^2}{k^2} =  \Omega\left(\frac{\log n}{n}\right) \,, 
\]
then the SDP-relaxed problem \eqref{opt:primal} achieves exact inference, i.e., $Y = Y^\ast = Z^\ast Z^{\ast\top}$, with probability at least $1 - O\left(\frac{1}{n}\right)$.
\label{sdp:theorem}
\end{theorem}

\section{Additional Analysis}
In this section, for completeness, we also provide an information-theoretic lower bound on exact inference (i.e., the impossible regime), and we analyze when (nonconvex) maximum likelihood estimation is correct (i.e., the hard regime).
\subsection{Impossible Regime}
In this section we analyze the necessary conditions for exact inference of latent models. Our goal is to characterize the information-theoretic lower limit of any algorithm for inferring the true labels $Z^\ast$ in our model. More specifically, we would like to infer labels $\hat{Z}$ given the observation of the adjacency matrix $W$. Also note that we do not observe the collection of latent variables $\mtx{X}$. We present the following information-theoretic lower bound for our model.

\begin{claim}
Let $Z^\ast$ be the true assignment matrix sampled uniformly at random. In a latent model of $k$ clusters and $n$ entities, and with induced parameters $(p, q)$ as in Definition \ref{def:induced}, if 
\[
\frac{p}{k}\log(p/q) = O\left(\frac{1}{n}\right) \,, 
\]
then the probability of error $\Prob{\hat{Z} \neq Z^\ast} \geq 1/2$, for any algorithm that a learner could use for picking $\hat{Z}$. 
\label{fano:theorem}
\end{claim}
\subsection{Hard Regime with Maximum Likelihood Estimation}

In this section we analyze the conditions for exact inference of the true labels in latent models using nonconvex maximum likelihood estimation by solving optimization problem \eqref{opt:mle2}.
We call this the hard regime because without some convex relaxation, enumerating $Z$ takes $O(k^n)$ iterations. The problem can be rewritten in the following square matrix form:
\begin{align}
\maximize_{Y \in \mathcal{Y}} \qquad & \langle W, Y\rangle  \,,
\label{nphard:primal}
\end{align}
where 
\[
\mathcal{Y} = \lbrace ZZ^\top \mid Z \in \{0,1\}^{n\times k}, Z^\top \onevct_n = \frac{n}{k}  \onevct_k, Z\onevct_k = \onevct_n \rbrace \,,
\]
is the space of all feasible solutions.
We now state the conditions for exact inference of latent models using maximum likelihood estimation.
\begin{claim}
In a latent model of $k$ clusters and $n$ entities, and with induced parameters $(p, q)$ as in Definition \ref{def:induced}, if
\[
\frac{(p-q)^2}{k} = \Omega\left(\frac{\log n}{n}\right) \,, 
\]
then maximum likelihood estimation \eqref{nphard:primal} achieves exact inference, i.e., $Y = Y^\ast = Z^\ast Z^{\ast\top}$, with probability at least $1- O\left(\frac{1}{n}\right)$.
\label{hard:theorem}
\end{claim}

\section{Experiments}
We validate our theoretical findings through experiments. We run synthetic experiments for the latent space model, the exchangeable graph model, and the kernel latent variable model. We also test our algorithm in a real-world dataset in which assumptions might not necessarily hold. See Appendix for details.


\bibliography{0_main.bib}
\bibliographystyle{plain}

\onecolumn
\newpage
\appendix

\section{Proof of LCI Concentration Inequalities}
In this section we present the proof of LCI concentration inequalities used in the main paper.


\begin{proof}[Proof of Lemma \ref{thm:LCI_tail}]
Starting from the left-hand side, we have
\begin{align*}
\Prob[XY]{\sum_i (x_i - \mu_i) \geq \epsilon}
&= \Prob[XY]{\econst^{t\sum_i (x_i - \mu_i)} \geq \econst^{t\epsilon}}, \text{ for every $t > 0$} \\
&\leq \econst^{-t\epsilon} 
\cdot \Expect[XY]{\econst^{t\sum_i (x_i - \mu_i)}} \\
&= \econst^{-t\epsilon} 
\cdot \Expect[Y]{\Expect[X]{\econst^{t\sum_i (x_i - \mu_i)} \middle| Y}} \\
&= \econst^{-t\epsilon} \cdot \Expect[Y]{\prod_i \Expect[x_i]{\econst^{t (x_i - \mu_i)} \middle| Y}} \\
&\leq \econst^{-t\epsilon} \cdot \Expect[Y]{\prod_i \exp\left(\tfrac{1}{2} t^2 \sigma_i^2 \right)} \\
&= \exp\left( -t\epsilon + \tfrac{1}{2}t^2 \sum_i \sigma_i^2 \right), \text{ let $t = \frac{\epsilon}{\sum_i \sigma_i^2}$} \\
&= \exp\left( -\frac{\epsilon^2}{2\sum_i \sigma_i^2} \right) \,.
\end{align*}
The second line follows from Markov's inequality, the third line follows from the law of total expectation, the fourth line follows from the LCI assumption, and the fifth line follows from the assumption $\Expect[x_i]{\econst^{t (x_i - \mu_i)} \mid Y} \leq \econst^{t^2 \sigma_i^2 / 2}$.
This completes the proof.
\end{proof}

\begin{proof}[Proof of Corollary \ref{cor:LCI_hoeffding}]
By Hoeffding's lemma we have 
$\Expect[x_i]{\econst^{t (x_i - \mu_i)} \middle| Y} \leq \econst^{t^2 (b_i-a_i)^2 / 8}$. Setting $\sigma_i^2 = (b_i-a_i)^2/4$ in the statement of Theorem \ref{thm:LCI_tail} leads to the desired result.
\end{proof}


\begin{proof}[Proof of Corollary \ref{cor:LCI_bern}]
For any single $x_i$, by Taylor expansion and the assumption of $\abs{x_i} \leq M$ and $\Var[x_i]{x_i \mid Y} \leq \nu_i^2$, we have
\[
\Expect[x_i]{\econst^{tx_i} \middle| Y}
= \Expect[x_i]{1 + tx_i + \sum_{p=2}^{\infty} \frac{t^p x_i^p}{p!} \middle| Y}
\leq \exp\left(\frac{\nu_i^2 (\econst^{tM} - tM - 1)}{M^2}  \right) \,,
\]
for any $t > 0$.
Setting $\sigma_i^2 = \frac{2\nu_i^2 (\econst^{tM} - tM - 1)}{t^2 M^2}$ in the statement of Theorem \ref{thm:LCI_tail} leads to the desired result.
\end{proof}


Before we present the proof of LCI matrix Bernstein ineqauality, we first introduce the proof of Lemma \ref{corollary_3.7}, which is motivated by \citep{tropp2012user}.

\begin{proof}[Proof of Lemma \ref{corollary_3.7}]
Starting from the left-hand side, we have
\begin{align*}
\Prob[XY]{\lmax\left(\sum_i X_i\right) \geq \epsilon}
&= \Prob[XY]{\econst^{\theta \lmax\left(\sum_i X_i\right)} \geq \econst^{\theta \epsilon}}, \text{ for every $\theta > 0$} \\
&\leq \econst^{-\theta \epsilon} 
\cdot \Expect[XY]{\econst^{\theta \lmax\left(\sum_i X_i\right)}} \\
&\leq \econst^{-\theta \epsilon} 
\cdot \Expect[XY]{\tr(\econst^{\theta \sum_i X_i})} \\
&= \econst^{-\theta \epsilon} 
\cdot \Expect[Y]{\Expect[X]{\tr(\econst^{\theta \sum_i X_i}) \middle| Y}} \\
&\leq \econst^{-\theta \epsilon} 
\cdot \Expect[Y]{\tr  \exp \left(\sum_i \log \Expect[X_i]{\econst^{\theta X_i} \middle| Y} \right)  } \\ 
&\leq \econst^{-\theta \epsilon} \cdot \Expect[Y]{\tr \exp\left( g(\theta) \sum_i A_i \right)  }  \\
&= \econst^{-\theta \epsilon} \cdot \tr \exp\left( g(\theta) \sum_i A_i \right)   \\
&\leq d\econst^{-\theta \epsilon}  \cdot \lmax \left( \exp\left( g(\theta) \sum_i A_i \right) \right)  \\
&= d\econst^{-\theta \epsilon}  \cdot \exp\left( g(\theta) \lmax\left(\sum_i A_i \right)\right)  \,.
\end{align*}
The second line follows from Markov's inequality, the third line follows from the spectral mapping theorem, the fourth line follows from the law of total expectation, the fifth line follows from the LCI assumption and the fact that the matrix cumulant generating functions are subadditive, and the sixth line follows from the assumption  
$
\Expect[X_i]{\econst^{\theta X_i} \mid Y} \preceq \econst^{g(\theta) \cdot A_i}
$. 
This completes the proof.
\end{proof}

We now present the proof of LCI matrix Bernstein inequality.

\begin{proof}[Proof of Corollary \ref{thm:LCI_mtxBern}]
In this proof we assume $R =1$ for simplicity. The general case follows by scaling the corresponding terms.

For any single $X_i$, by Taylor expansion and the assumption of $\Expect[X_i]{X_i \mid Y} = \zeromtx$, we have
\[
\Expect[X_i]{\econst^{\theta X_i} \middle| Y}
= \Expect[X_i]{\Imtx + \theta X_i + \sum_{p=2}^{\infty} \frac{\theta^p X_i^p}{p!} \middle| Y}
\preceq \exp\left((\econst^{\theta} - \theta - 1) \cdot \Expect[X_i]{X_i^2 \middle| Y}   \right) \,,
\]
for any $\theta > 0$.
Then by Lemma \ref{corollary_3.7} we have
\begin{align*}
\Prob[XY]{\lmax\left(\sum_i X_i\right) \geq \epsilon}
&\leq d\econst^{-\theta \epsilon}  \cdot \exp\left( (\econst^\theta - \theta - 1) \cdot \lmax\left(\sum_i \Expect[X_i]{X_i^2 \middle| Y} \right)\right) \\
&\leq d\econst^{-\theta \epsilon}  \cdot \exp\left( (\econst^\theta - \theta - 1) \cdot \sigma^2 \right) \,.
\end{align*}
Setting $\theta = \log(1+\epsilon/\sigma^2)$ completes the proof.
\end{proof}
\section{Proofs for Polynomial-Time Regime with Semidefinite Programming}

\begin{proof}[Proof of Lemma \ref{lemma:dual}]
We define the Lagrangian variables $a_1, \dots, a_n \in \R, v_1,v_2 \in \R^n, \Lambda \succeq_\Snp 0, \Gamma \succeq_\Rnp 0$ for the constraints in \eqref{opt:primal} respectively. Then Lagrangian of \eqref{opt:primal} is
\begin{align*}
L(Y, a_1, \dots, a_n, v, \Lambda, \Gamma) 
= 
&\inprod[-W+\diag(a_1,\dots,a_n) + \onevct_n v_1^\top + v_2\onevct_n^\top - \Lambda - \Gamma]{Y} \\
&- \sum_i a_i - m (v_1+v_2)^\top \onevct_n \,.
\end{align*}
For simplicity we denote $A = \diag(a_1,\dots,a_n)$ to be a diagonal matrix. By the KKT stationarity condition and dual feasibility, we have
\begin{equation}
-W + A + \onevct_n v_1^\top + v_2\onevct_n^\top - \Gamma = \Lambda \succeq_\Snp 0\,.
\label{opt:stationarity}    
\end{equation}
Note that in the equation above, positive semidefiniteness requires symmetry. Thus we set $v:= v_1 = v_2$, and we require $\Gamma$ to be symmetric. Then we obtain the dual objective function $\tr(A) + 2mv^\top \onevct_n$. 

We now look at the remaining constraints. The KKT complementary slackness condition requires that 
\begin{equation}
\inprod[\Lambda]{Y} = 0
\label{opt:CSa}
\end{equation}
and
\begin{equation}
\Gamma_{ij} Y_{ij} = 0
\label{opt:CSb}
\end{equation}
for every $i$ and $j$. We want to highlight that \eqref{opt:CSa} is equivalent to $\Lambda Y = 0$, given that both matrices are positive semidefinite. Since $Y^\ast = Z^\ast Z^{\ast\top}$, this implies that for the optimal solution $Y^\ast$, every $z^{\ast(i)}$ is an eigenvector of $\Lambda$ with an eigenvalue of $0$. Furthermore $\eqref{opt:CSb}$ implies that $\Gamma_{S_i} = 0$ for all $i \in [k]$, because $Y^\ast_{S_i}$ is an all-one submatrix.
\end{proof}


\begin{proof}[Proof of Lemma \ref{lemma:certificates}]
Strong duality requires that the optimal primal and dual objective values are equal. In other words, the objective value of problem \eqref{opt:primal} and \eqref{opt:dual} should match. 
Note that the optimal primal solution $Y^\ast$ can be decomposed as $Y^\ast = Z^\ast Z^{\ast\top}$. Thus the primal objective function can be rewritten as 
\begin{align*}
\inprod[W]{Y^\ast} 
= \inprod[W]{Z^\ast Z^{\ast\top}}
= \sum_{i=1}^n \sum_{j=1}^n W_{ij} z_i^{\ast\top} z_j^\ast \,.
\end{align*}
On the other hand, the dual objective function is equal to
\[
\tr(A) + 2m v^\top \onevct_n  
= \tr(D) - mn\phi + mn \phi
=  \sum_{i=1}^n d_i \,.
\]
Recall that $d_i = \sum_{j=1}^n W_{ij} z_i^{\ast\top} z_j^\ast$ and $D = \diag(d_1, \dots, d_n)$. One can see that by setting $a_i = d_i - m \phi$, or $A = D - m \phi \Imtx$, the duality gap is closed. 
One may notice that the choice of $\phi$ does not change the objective values here. For the sole purpose of strong duality, $\phi$ is an arbitrary constant that will be determined later.
\end{proof}


\begin{proof}[Proof of Lemma \ref{lemma:df}]
This directly follows from the constraint in \eqref{opt:dual}, by plugging in the construction of $v, A$ and $\Gamma$.
\end{proof}


\begin{proof}[Proof of Lemma \ref{lemma:uniqueness}]
Again we use $Y^\ast = Z^\ast Z^{\ast\top}$ to denote the optimal primal solution. Since  $\Lambda$ and $Y^\ast$ are both positive semidefinite, the KKT complementary slackness condition \eqref{opt:CSa} is equivalent to $\Lambda Y^\ast = \Lambda Z^\ast Z^{\ast\top}= 0$, which implies that every $z^{\ast(i)}$ is an eigenvector of $\Lambda$ with an eigenvalue of $0$. Condition  \eqref{eq:eigen_condition} further requires that $\{z^{\ast(i)}\}_{i=1}^k$ spans the whole null space of $\Lambda$. As a result, any optimal primal solution $Y$ needs to be a multiple of $Y^\ast$. Since $Y_{ii} = 1$, the choice of $Y$ is unique.
\end{proof}


\begin{proof}[Proof of Lemma \ref{theorem:sdp_deterministic}]
This directly follows from Lemma \ref{lemma:df} and \ref{lemma:uniqueness}.
\end{proof}


\begin{proof}[Proof of Lemma \ref{lemma:choice_phi}]
Motivated by \citep{amini2018semidefinite}, in the following proof we introduce the notation $\bar{d}(S_l, S_r)$ to denote the average degree of connectivity between cluster $l$ and $r$. In other words, we have $\bar{d}(S_l, S_r) := \frac{1}{m} \sum_{i\in S_r} d_i(S_l) = \frac{1}{m} \sum_{j\in S_l} d_j(S_r)$.
Note that dual feasibility condition \eqref{eq:entrygeq0} is satisfied, if for every $i \in S_r, j \in S_l$, we have
\[
m\phi \geq d_i(S_l) + d_j(S_r) - \bar{d}(S_l, S_r) \,.
\]
By definition, dividing both sides by $m$, this is equivalent to
\[
\phi + (PW_{S_l S_r}P - W_{S_l S_r})_{ij} \geq 0 \,.
\]

One may note that each random variable $d_i(S_l)$ is the summation of $m$ LCI random variables given $X$, with the expectation $\Expect[WX]{d_i(S_l)} = mq$. Using LCI Hoeffding's inequality, we obtain
\[
\Prob[WX]{\abs{d_i(S_l) - \Expect[WX]{d_i(S_l)}} \geq t} \leq 2\exp\left( - \frac{2t^2}{m}\right) \,.
\]
Taking a union bound for all $S_l, S_r$ and $i\in S_r$ gives us 
\begin{align*}
\Prob[WX]{\max_i \abs{d_i(S_l) - \Expect[WX]{d_i(S_l)}} \geq t} 
&\leq 2m \binom{k}{2}\exp\left( - \frac{2t^2}{m}\right) \\
&\leq 2nk \exp\left( - \frac{2t^2}{m}\right) \\
& \leq O(n^{-1}) \,,
\end{align*}
where the last inequality holds if $t\geq O\left(\sqrt{\frac{n\log n}{k}}\right)$. 

Note that by definition, the average degree $\bar{d}(S_l, S_r)$ is always bounded between the minimum and the maximum of $d_i(S_l)$ and $d_j(S_r)$. Then with probability at least $1 - O(n^{-1})$, it follows that
\begin{align*}
&\quad \abs{d_i(S_l) + d_j(S_r) - \bar{d}(S_l, S_r) - \Expect[WX]{d_i(S_l) + d_j(S_r) - \bar{d}(S_l, S_r)}} \\
& = \abs{d_i(S_l) + d_j(S_r) - \bar{d}(S_l, S_r) - mq} \\
&\leq 3t \,.
\end{align*}
Dividing both sides by $m = \frac{n}{k}$, this is equivalent to 
\[
\abs{(PW_{S_l S_r}P - W_{S_l S_r})_{ij} - q}
\leq \frac{3kt}{n} 
\leq O\left(\sqrt{\frac{k\log n}{n}}\right) \,.
\]
Thus, by setting $\phi \geq q+O\left(\sqrt{\frac{k\log n}{n}}\right)$, nonnegativity is satisfied with probability at least $1 - O(n^{-1})$.
\end{proof}


\begin{proof}[Proof of Lemma \ref{lemma:sdp_expectation}]
Here we look at the expectation of $\Lambda$. Note that
\begin{align*}
\lambda_{k+1} (\Expect[WX]{\Lambda})
&=\lambda_{k+1} \left(\Expect[WX]{D-m\phi \Imtx - W + \phi \onemtx - \Gamma}\right) \\
&= \min_{\substack{u\perp z^{\ast(1)}, \dots, z^{\ast(n)} \\ u\in \sphere}}
u^\top \left(m(p - \phi)\Imtx - \Expect[WX]{\Gamma}\right) u \\
&= m (p-\phi) - \min_{\substack{u\perp z^{\ast(1)}, \dots, z^{\ast(n)} \\ u\in \sphere}}
u^\top \Expect[WX]{\Gamma} u \\
&=  m (p-\phi) - 
\min_{\substack{u_i, u_j\perp \onevct_{{m}} \\ \sum_i \norm{u_i}^2 = 1}}
\sum_{i\neq j}
u_i^\top \Expect[WX]{\Gamma_{S_i S_j}} u_j \,.
\end{align*}
Note that for each summand above, we have
\begin{align*}
u_i^\top \Expect[WX]{\Gamma_{S_i S_j}} u_j
&= u_i^\top \left(\phi\onemtxm + P\Expect[WX]{W_{S_i S_j}}P + \Expect[WX]{W_{S_i S_j}}\right) u_j \\
&= u_i^\top \left(\phi\onemtxm + qP\onemtxm P + q\onemtxm \right) u_j \\
&= 0 \,,
\end{align*}
given that $u_i$'s are orthogonal to $\onevct$. Thus we obtain 
\begin{equation}
\lambda_{k+1} \left(\Expect[WX]{\Lambda}\right) = m (p-\phi) \,.
\end{equation}
\end{proof}


\begin{proof}[Proof of Lemma \ref{lemma:dw}]
Starting from \eqref{eq:eigenk+1}, we have 
\begin{align}
\lambda_{k+1} (\Lambda)  \nonumber
&\geq \lambda_{k+1} (\Expect[WX]{\Lambda}) + \lambda_{k+1} (\Lambda - \Expect[WX]{\Lambda}) \nonumber \\
&\geq  m(p-\phi) \nonumber\\
&\quad+ \lambda_{k+1} \left(D-\Expect[WX]{D}\right) \label{eq:D_total} \\
&\quad+ \lambda_{k+1} \left(-W+\Expect[WX]{W}\right) \label{eq:W_total} \\
&\quad+ \lambda_{k+1} \left(-\Gamma+\Expect[WX]{\Gamma}\right) \,. \label{eq:Gamma_total} 
\end{align}

Regarding \eqref{eq:D_total}, note that $D$ is a diagonal matrix. As a result, $\lambda_{k+1} \left(D-\Expect[WX]{D}\right) \geq \min_i (d_i - \Expect[WX]{d_i}) $.

Regarding \eqref{eq:W_total}, it follows that $\lambda_{k+1} \left(-W+\Expect[WX]{W}\right) \geq -\lmax\left(W-\Expect[WX]{W}\right)$.

Regarding \eqref{eq:Gamma_total}, we have
\begin{align*}
\lambda_{k+1} \left(-\Gamma + \Expect[WX]{\Gamma}\right) 
&= \min_{\substack{u_i, u_j\perp \onevct_{m} \\ \sum_i \norm{u_i}^2 = 1}}
\sum_{i\neq j}
u_i^\top \left(-\Gamma_{S_i S_j} + \Expect[WX]{\Gamma_{S_i S_j}} \right) u_j \,.
\end{align*}
Note that 
\begin{align*}
u_i^\top \left(-\Gamma_{S_i S_j} + \Expect[WX]{\Gamma_{S_i S_j}}\right)u_j 
&= u_i^\top \left(-\phi\onemtxm - PW_{S_i S_j}P + W_{S_i S_j} + \phi\onemtxm - q \onemtxm\right)u_j \\
&= u_i^\top \left(- PW_{S_i S_j}P + W_{S_i S_j} - q\onemtxm\right)u_j \\
&= 0 \,,
\end{align*}
given that $u_i$'s are orthogonal to $\onevct$. Thus $\lambda_{k+1} \left(-\Gamma + \Expect[WX]{\Gamma}\right) = 0$.

Combining the results above, it is sufficient to prove that
\begin{align}
\lambda_{k+1} (\Lambda) 
\geq m(p-\phi) + \min_i (d_i - \Expect[WX]{d_i}) -\lmax\left(W-\Expect[WX]{W}\right) \geq 0 \,.
\end{align}
This gives us the result in the statement.
\end{proof}


\begin{proof}[Proof of Theorem \ref{sdp:theorem}]
Our proof relies on the use of LCI concentration inequalities. 
First we show that \eqref{eq:d_con} holds with high probability. Note that, for any fixed latent variable $X$ and any $i\in [n]$, we have $\Prob[WX]{d_i - \Expect[WX]{d_i \middle| X} \leq -t } \leq \exp(-2t^2 / n)$ by LCI Hoeffding's inequality. By a union bound, it follows that 
\begin{align*}
&\Prob[WX]{\min_i (d_i - \Expect[WX]{d_i}) \leq -t} 
\leq n\exp(-2t^2 / n) \,.
\end{align*}
Setting $t = \frac{1}{2}m(p-\phi)$, we obtain
\begin{align*}
\Prob[WX]{\min_i (d_i - \Expect[WX]{d_i}) \leq -\frac{1}{2}m(p-\phi)} 
\leq n\exp\left(- \frac{n(p-\phi)^2}{2k^2} \right) 
\leq  c_1 n^{-1} \,,
\end{align*}
where the last inequality holds given that $\left(\frac{p-\phi}{k}\right)^2 \geq 4 \frac{\log n}{n} - 2 \frac{\log c_1}{n}$.

Next we show that \eqref{eq:w_con} holds with high probability, and we use LCI Bernstein inequality in our proof. In this part we denote $\Bar{W}_{ij} := W_{ij} - \Expect[WX]{W_{ij}}$, and $\delta_{ij}$ to be the matrix with $1$ in entry $(i,j), (j,i)$, and $0$ everywhere else. Note that $\delta_{ij}^2$ is a matrix with $1$ in entry $(i,i),(j,j)$, and $0$ everywhere else. 
Furthermore we define the matrix $\Delta_{ij} := \Bar{W}_{ij} \delta_{ij}$. One can note that $\Delta_{ij}$'s are LCI random matrices given $X$, with the maximum eigenvalue bounded above by $1$.
Also note that for any given $X$, we have $\Expect[W]{\Delta_{ij} \middle| X} = \zeromtx$. By our construction, it follows that $\sum_{i<j} \Delta_{ij} = W - \Expect[WX]{W}$. Thus for any given $X$, it follows that $\norm{\sum_{i<j} \Expect[W]{\Delta_{ij}^2  \middle| X}} = \norm{\sum_{i<j} \delta_{ij}^2 \Expect[WX]{\Bar{W}_{ij}^2  \middle| X}} \leq n-1$.
Then applying the LCI matrix Bernstein inequality, we obtain
\begin{align*}
&\Prob[WX]{\lmax(W - \Expect[WX]{W}) \geq t} 
\leq  n\exp\left(-\frac{t^2/2}{n + t/3}\right) \,.
\end{align*}
Setting $t = \frac{1}{2}m(p-\phi)$, we obtain
\begin{align*}
\Prob[WX]{\lmax(W - \Expect[WX]{W}) \geq \frac{1}{2}m(p-\phi)}
\leq n\exp\left(-\frac{n (p-\phi)^2}{16k^2}\right) 
\leq c_2 n^{-1} \,,
\end{align*}
where the last inequality holds given that $\left(\frac{p-\phi}{k}\right)^2 \geq 32 \frac{\log n}{n} - 16 \frac{\log c_2}{n}$.

Combining the results above, the probability of $\lambda_{k+1}(\Lambda)$ being greater than zero is at least $1-(c_1+c_2)n^{-1}$, as long as $\left(\frac{p-\phi}{k}\right)^2 \geq 32 \frac{\log n}{n}$. The last remaining task is to take $\phi$ into account. By Lemma \ref{lemma:choice_phi}, setting $\phi = q+c\sqrt{\frac{k\log n}{n}}$ for some constant $c$ gives us
\[
\left(p-q-c\sqrt{\frac{k\log n}{n}}\right)^2 \geq 32k^2 \frac{\log n}{n} \,.
\]
Simplification leads to
\[
(p-q)^2 \geq \left(2c(p-q) + 32k\sqrt{\frac{k\log n}{n}}\right) \sqrt{\frac{k\log n}{n}} \,.
\]
To further simplify the bound above we consider two cases. If $2c(p-q) \geq 32k\sqrt{\frac{k\log n}{n}}$, a sufficient condition will be $(p-q)^2 \geq 16c^2 k \frac{\log n}{n}$. On the other hand if $2c(p-q) \leq 32k\sqrt{\frac{k\log n}{n}}$, a sufficient condition will be $(p-q)^2 \geq 64k^2 \frac{\log n}{n}$. 
Thus for either case, $\frac{(p-q)^2}{k^2} \geq  c'\frac{\log n}{n}$, for some large constant $c'$, is a sufficient condition. This completes our proof.
\end{proof}
\section{Proofs for Additional Analysis}

\subsection{Proof of Claim \ref{fano:theorem}}
In the following proof, we use notation $\mathcal{Z}$ to denote the space of feasible solutions. Mathematically, we have the following definition
\[
\mathcal{Z} = \lbrace Z \mid Z \in \{0,1\}^{n\times k}, Z^\top \onevct_n = \frac{n}{k}  \onevct_k, Z\onevct_k = \onevct_n \rbrace \,, 
\]
and we assume that the groundtruth $Z^\ast$ is sampled uniformly at random from $\mathcal{Z}$.

\begin{proof}
First we characterize the mutual information between the true labels $Z^\ast$ and the observed matrix $W$. Using the pairwise KL-based bound \citep{yu1997assouad}, we obtain
\begin{align}
I(Z^\ast, W) \nonumber  
&\leq  \frac{1}{\abs{\mathcal{Z}}^2} \sum_{Z, Z' \in \mathcal{Z}} \kl(\Prob[\mtx{WX}]{W \mid Z} \mid \Prob[\mtx{WX}]{W \mid Z'}) \nonumber\\
&\leq \max_{Z, Z' \in \mathcal{Z}} \kl(\Prob[\mtx{WX}]{W \mid Z} \mid \Prob[\mtx{WX}]{W \mid Z'}) \nonumber\\
&\leq \binom{n}{2} \max_{Z, Z'} \kl(\Prob[\mtx{WX}]{W_{ij} \mid z_i, z_j} \mid \Prob[\mtx{WX}]{W_{ij} \mid z'_i, z'_j}) \nonumber\\
&\leq \binom{n}{2} \sum_{W_{ij}} \Prob[\mtx{WX}]{W_{ij} \mid z_i, z_j} \log \frac{\Prob[\mtx{WX}]{W_{ij} \mid z_i, z_j}}{\Prob[\mtx{WX}]{W_{ij} \mid z'_i, z'_j}} \nonumber\\
&\leq \binom{n}{2}  \Prob[\mtx{WX}]{W_{ij} = 1 \mid z_i, z_j} \log \frac{\Prob[\mtx{WX}]{W_{ij} = 1 \mid z_i, z_j}}{\Prob[\mtx{WX}]{W_{ij} = 1 \mid z'_i, z'_j}} \nonumber\\
&\leq  \binom{n}{2} \Expect[\mtx{X}]{f(\vct{x_i}, \vct{x_j}) \mid z_i = z_j} 
\log \frac{\Expect[\mtx{X}]{f(\vct{x_i}, \vct{x_j}) \mid z_i = z_j}}{\Expect[\mtx{X}]{f(\vct{x_i}, \vct{x_j}) \mid z_i \neq z_j}} \nonumber \\
&= \binom{n}{2} p \log \frac{p}{q} \nonumber\\
&\leq \frac{n^2}{2} p \log \frac{p}{q} \,, \nonumber
\end{align}
where $\kl(\cdot \mid \cdot)$ denotes the KL-divergence between two probability distributions. Then we can apply Fano's inequality \citep{cover2012elements}. For any predicted labels $\hat{Z}$, we have
\begin{align*}
\Prob{\hat{Z} \neq Z^\ast}
&\geq 1 - \frac{I(Z^\ast, W) + \log 2}{\log \abs{\mathcal{Z}}} \,.
\end{align*}

By definition of $\mathcal{Z}$ and counting, it follows that
\begin{align*}
\abs{\mathcal{Z}} &= \frac{n!}{k!(m!)^k} \,.
\end{align*}

Note that $\sqrt{n}(n/\econst)^n \leq n! \leq \econst\sqrt{n}(n/\econst)^n$. It follows that
\begin{align*}
\abs{\mathcal{Z}} 
\geq &\frac{\sqrt{n}(n/\econst)^n}{\econst \sqrt{k} (k/\econst)^k \econst^k \sqrt{m}^k (m/\econst)^{mk}} 
\geq  n^{\frac{1-k}{2}} k^{n-\frac{1+k}{2}} \econst^{k-1},
\end{align*}
which indicates that
\begin{align*}
\log \abs{\mathcal{Z}} 
\geq \frac{1}{2} \left(nk+n-k^2+k\right) 
\geq \frac{1}{4}nk \,,
\end{align*}
and the last inequality holds under the mild assumption of $n \geq 2k$.

Finally, by Fano's inequality, for the probability of error to be at least $1/2$, it is sufficient to require the lower bound to be greater than $1/2$. Hence 
\begin{align}
\Prob{\hat{Z} \neq Z^\ast}
&\geq 1 - \frac{I(Z^\ast, W) + \log 2}{\log \abs{\mathcal{Z}}} \nonumber \\
&\geq 1 - \frac{\frac{n^2}{2} p \log (p/q) + \log 2}{\frac{1}{4}nk} \nonumber \\
&= 1 - 2np\frac{\log (p/q)}{k} - \frac{4}{nk} \nonumber \\
&= \frac{nk-4}{nk} - 2np\frac{\log (p/q)}{k} \nonumber \\
&\geq \frac{1}{2} \,, \nonumber
\end{align}
and the last inequality holds provided that $\frac{p}{k}\log(p/q) \leq \frac{1}{8n}$ and $nk \geq 8$.
\end{proof}

\subsection{Proof of Claim \ref{hard:theorem}}
In the following proof we define $d(\my) = \langle \mya, \mya - \my \rangle \geq 0$. Before we start our proof we first present the following result.

\begin{lemma}[Lemma 1.1, \citep{chen2014statistical}]
For each $t\in [2(n/k - 1), n^2/k]$, we have
\[
\abs{\{\my \in \mathcal{Y}: d(\my) = t\}} 
\leq
\frac{25k^2 t^2}{n^2} n^{20kt/n} \,.
\]
\label{lemma:chen}
\end{lemma}

Our proof consists of two steps. We first show the deterministic condition for problem \eqref{nphard:primal} to succeed, and then derive the statistical condition by bounding $\mw$ from its expectation $\Expect[\mall]{\mw}$. We present the following lemma.
\begin{lemma}
If the following condition
\begin{align}
(p-q)d(\my) + \langle \mw - \Expect[\mall]{\mw}, \mya - \my \rangle &> 0 \label{hard:mainA} 
\end{align}
\label{hard:lemma_decomposition}
holds, then maximum likelihood estimation \eqref{nphard:primal} achieves exact inference.
\end{lemma}

\begin{proof}
To prove problem \eqref{nphard:primal} returns the optimal solution, it is sufficient to prove that for every $\my \neq \mya$, $\langle \mw, \mya - \my \rangle$ is strictly positive. Note that 
\begin{align}
\langle \mw, \mya - \my \rangle
=  \langle \mw - \Expect[\mall]{\mw}, \mya - \my \rangle 
+ \langle \Expect[\mall]{\mw}, \mya - \my \rangle 
\label{nphard:conc}
\end{align}

Regarding the last term in \eqref{nphard:conc}, note that $\Expect[\mall]{\mw} = q\onemtx + (p-q)\mya $. Given the fact that $\normf{\mya} = \normf{\my}$ for every $\my \in \mathcal{Y}$, we have 
\begin{align*}
\langle \Expect[\mall]{\mw}, \mya - \my \rangle 
&=  (p-q) d(\my) . 
\end{align*}
\end{proof}

We now present the proof of Theorem \ref{hard:theorem}.
\begin{proof}
To show that \eqref{hard:mainA} holds with high probability, we use LCI Bernstein inequality in our proof. For any fixed collection of latent variables $\mtx{X}$ and any $i\neq j$, $(\mw_{ij} - \Expect[\mall]{\mw_{ij}})(\mya_{ij} - \my_{ij})$ is a Bernoulli random variable centered at $0$, bounded between $-1$ and $1$, with a variance bounded above by $\frac{1}{4}$. Thus $\langle \mw - \Expect[\mall]{\mw}, \mya - \my \rangle  =   2\sum_{i < j} (\mw_{ij} - \Expect[\mall]{\mw_{ij} })(\mya_{ij} - \my_{ij}) $ is the summation of $\frac{1}{2} d(\my)$ LCI random variables given $X$. LCI Bernstein inequality implies
\[
\Prob[\mall]{\sum_{i < j} (\mw_{ij} - \Expect[\ma]{\mw_{ij}})(\mya_{ij} - \my_{ij}) \leq - t} 
\leq
\exp\left(-\frac{t^2}{2t/3 +d(\my) / 4}\right),
\]
for every $t > 0$. 

Setting $t = (p-q) d(\my)$, it follows that
\begin{align}
\Prob[\mtx{WX}]{\langle \mw - \Expect[\ma]{\mw \mid \mx}, \mya - \my \rangle \leq -(p-q) d(\my)} 
\leq 
\exp\left(-\frac{1}{7}  (p-q)^2 d(\my)\right). \nonumber
\end{align}

By a union bound we obtain
\begin{align*}
&\Prob[\mall]{\exists \my \neq \mya: \langle \mw - \Expect[\ma]{\mw \mid \mx}, \mya - \my \rangle \leq - \frac{1}{4}(p-q) d(\my)} \\
\leq & \sum_{t = 2(n/k - 1)}^{n^2/k} \abs{\{\my\in\mathcal{Y}: d(\my) = t\}} \exp\left(-\frac{1}{7} (p-q)^2  t\right) \\
\leq & \sum_{t = 2(n/k - 1)}^{n^2/k} \frac{25k^2 t^2}{n^2} n^{20kt/n} \exp\left(-\frac{1}{7} (p-q)^2  t\right) \\
\leq & \sum_{t = 2(n/k - 1)}^{n^2/k} 25{n^2} n^{20kt/n} \exp\left(-\frac{1}{7} (p-q)^2  t\right) \\
\leq & \sum_{t = 2(n/k - 1)}^{n^2/k} 25{n^2} \frac{k}{25} n^{-5} \\
\leq & n^{-1} ,
\end{align*}
where the third line follows Lemma \ref{lemma:chen}, and the second to last inequality holds given that $\frac{(p-q)^2}{k} \geq 175 \frac{\log n}{n}$.


Finally applying Lemma \ref{hard:lemma_decomposition}, the probability of $\langle \mw, \mya - \my \rangle$ being greater than zero is at least $1 - n^{-1}$. This completes our proof.
\end{proof}


\section{Experiments}
In this section, we validate our theoretical findings through synthetic experiments.
Here we compare the theoretic exact inference condition suggested by our SDP analysis, and the experimental results of exact inference using CVX \citep{cvx, gb08} to solve the SDP problem. We run synthetic experiments on four models: latent space model with three clusters, latent space model with two clusters, exchangeable graph model with two clusters, and kernel latent variable model with two clusters. 

\textbf{Latent space model with three clusters (Fig. \ref{fig:lsm3cluster}).} We pick $\mathcal{X} = \R^3$ as the latent domain. We fix the number of entities $n$ to be $30$. 
We generate $Z^\ast$ by randomly assigning entities to three groups of equal size. 
We generate the latent variables using Gaussian distributions, such that $P_{1} = \mathcal{N}_3((\norm{\mu},0,0), \sigma^2 \Imtx)$, $P_{2} = \mathcal{N}_3((0,\norm{\mu},0), \sigma^2 \Imtx)$, $P_{3} = \mathcal{N}_3((0,0,\norm{\mu}), \sigma^2 \Imtx)$, and $f(\vct{x}, \vct{x}') = \exp(-\normsq{\vct{x} - \vct{x'}})$.
The parameters in our simulations are $\norm{\vct{\mu}}$ and $\sigma$.
Each entry $\mtx{W}_{ij}$ follows Bernoulli distribution with probability $f(\vct{x}_i, \vct{x}_j)$. 
For each pair of $\norm{\vct{\mu}}$ and $\sigma$, we count: a) how many times (out of $10$) the fourth smallest eigenvalue of $\Lambda$ is greater than zero, and b) how many times (out of $5$) CVX returns the correct $\my = \mya = Z^\ast Z^{\ast\top}$. This allows us to compute an empirical probability of success for the statistical condition $\lambda_{k+1}(\Lambda) > 0$ and CVX, respectively.
Our experiments show that if the fourth smallest eigenvalue is strictly positive, then exact inference can be performed efficiently by semidefinite programming. 

\textbf{Latent space model with two clusters (Fig. \ref{fig:lsm}).} We pick $\mathcal{X} = \R^2$ as the latent domain. We fix the number of entities $n$ to be $150$.
Note that in the two cluster case, we can let the group assignment matrix $Z$ become a vector by using the $\{+1,-1\}$ encoding. We generate $Z^\ast$ by randomly assigning $n/2$ entities to one group ($z^\ast_i = 1$), and $n/2$ entities to the other group ($z^\ast_i = -1$). Since we are using the $\{+1,-1\}$ encoding, we only need to check the second smallest eigenvalue $\lsec(D-W) > 0$ as the sufficient condition. 
We generate the latent variables using Gaussian distributions, such that $P_{1} = \mathcal{N}_2(\vct{\mu}, \sigma^2 \Imtx)$, $P_{2} = \mathcal{N}_2(-\vct{\mu}, \sigma^2 \Imtx)$, where $\mathcal{N}$ denotes the Gaussian distribution. We also set $f(\vct{x}, \vct{x}') = \exp(-\normsq{\vct{x} - \vct{x'}})$. 
The parameters in our simulations are $\norm{\vct{\mu}}$ and $\sigma$.
Each entry $\mtx{W}_{ij}$ follows Bernoulli distribution with probability $f(\vct{x}_i, \vct{x}_j)$. 
For each pair of $\norm{\vct{\mu}}$ and $\sigma$, we count: a) how many times (out of $10$) the second smallest eigenvalue of $D-W$ is greater than zero, and b) how many times (out of $5$) CVX returns the correct $\my = \mya = Z^\ast Z^{\ast\top}$. Our experiments show that if the second smallest eigenvalue is strictly positive, then exact inference can be performed efficiently by semidefinite programming.

\textbf{Exchangeable graph model with two clusters (Fig. \ref{fig:hamming}).} We pick $\mathcal{X} = \{0,1\}^{32}$ as the latent domain. We fix the number of entities $n$ to be $150$. 
We generate $Z$ using the same method as in the latent space model with two clusters.
We generate the latent variables as follows: for every $x_i \in \{0,1\}^{32}$, its digits follow Bernoulli distribution with parameter $\alpha$, if entity $i$ is in the first group; its digits follow Bernoulli distribution with parameter $1-\alpha$, if entity $i$ is in the second group.  
We set $f(\vct{x}, \vct{x}') = \exp(-\norm{\vct{x} - \vct{x'}}_1 / \beta)$. 
The parameters in our simulations are $\alpha$ and $\beta$.
Each entry $\mtx{W}_{ij}$ follows Bernoulli distribution with probability $f(\vct{x}_i, \vct{x}_j)$. 
Our experiments show that if the second smallest eigenvalue is strictly positive, then exact inference can be performed efficiently by semidefinite programming. 

\textbf{Kernel latent variable model with two clusters (Fig. \ref{fig:kernel}).} We pick $\mathcal{X}$ to be the power set of $\{1,\dots,32\}$ as the latent domain. We fix the number of entities $n$ to be $150$. 
We generate $Z$ using the same method as in the latent space model with two clusters.
We generate the latent variables as follows: every $x_i$ is a subset of $\{1,\dots,32\}$. Each element $1$ through $16$ is in set $x_i$ with probability $p$ if entity $i$ is in the first group, and with probability $1-p$ if entity $i$ is in the second group.
Each element $17$ through $32$ is in set $x_i$ with probability $1-p$ if entity $i$ is in the first group, and with probability $p$ if entity $i$ is in the second group.
We set the kernel $K(\vct{x}, \vct{x}') = 2^{\abs{x \cap x'}}$, and $f(\vct{x}_i, \vct{x}_j) = (\log_2 K(x_i,x_j)) / 32$.
The parameters in our simulations are $p$ and $\alpha$.
Each entry $\mtx{W}_{ij}$ follows Beta distribution with parameters $(\alpha, \alpha\frac{1-f(\vct{x}_i, \vct{x}_j)}{f(\vct{x}_i, \vct{x}_j)})$. 
Our experiments show that if the second smallest eigenvalue is strictly positive, then exact inference can be performed efficiently by semidefinite programming. 

\begin{figure}[!ht]
    \centering
    \includegraphics[width=0.45\textwidth]{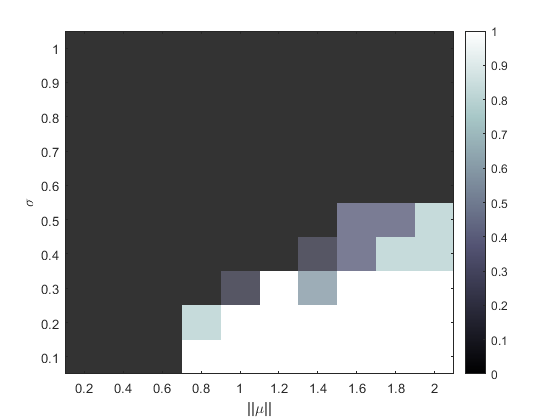}
    \includegraphics[width=0.45\textwidth]{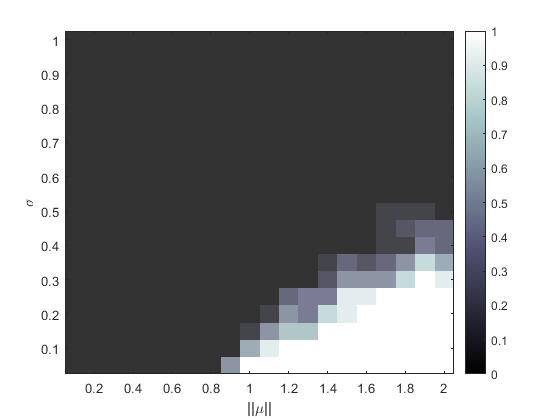}
    \caption{Latent space model with three clusters. \textbf{Left:} empirical probability of SDP returning the ground truth. \textbf{Right:} empirical probability of fourth smallest eigenvalue being positive. A brighter color indicates a higher percentage of successful trials.}
    \label{fig:lsm3cluster}
\end{figure}

\begin{figure}[!ht]
    \centering
    \includegraphics[width=0.45\textwidth]{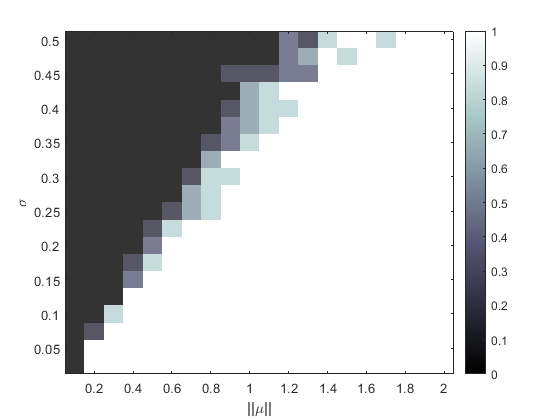}\includegraphics[width=0.45\textwidth]{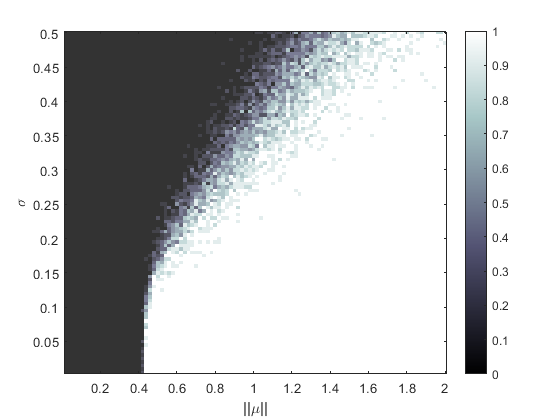}
    \caption{Latent space model with two clusters. \textbf{Left:} empirical probability of SDP returning the ground truth. \textbf{Right:} empirical probability of second smallest eigenvalue being positive. }
    \label{fig:lsm}
\end{figure}

\begin{figure}[!ht]
    \centering
    \includegraphics[width=0.45\textwidth]{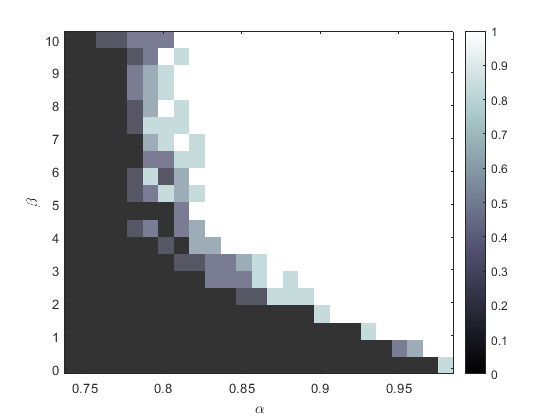}\includegraphics[width=0.45\textwidth]{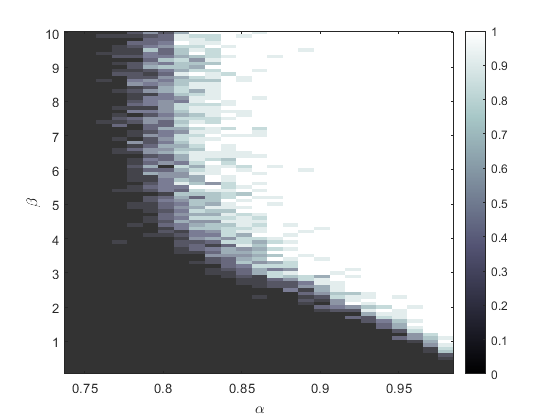}
    \caption{Exchangeable graph model with two clusters. \textbf{Left:} empirical probability of SDP returning the ground truth. \textbf{Right:} empirical probability of second smallest eigenvalue being positive.}
    \label{fig:hamming}
\end{figure}

\begin{figure}[!ht]
    \centering
    \includegraphics[width=0.45\textwidth]{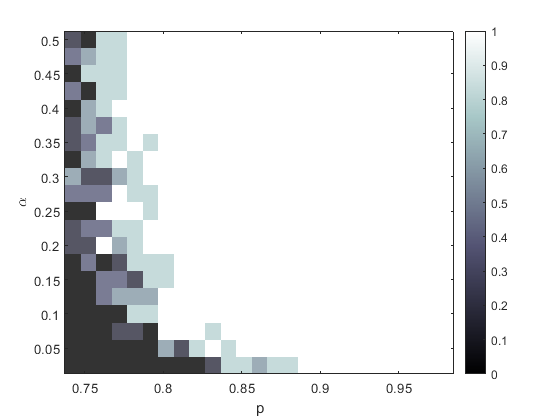}\includegraphics[width=0.45\textwidth]{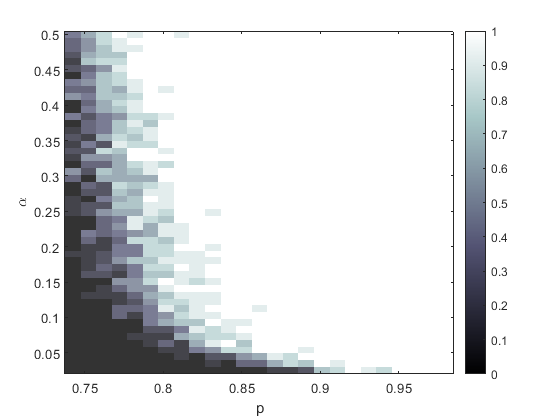}
    \caption{Kernel latent variable model with two clusters. \textbf{Left:} empirical probability of SDP returning the ground truth. \textbf{Right:} empirical probability of second smallest eigenvalue being positive.}
    \label{fig:kernel}
\end{figure}

\begin{figure}[!ht]
    \centering
    \includegraphics[width=0.45\textwidth]{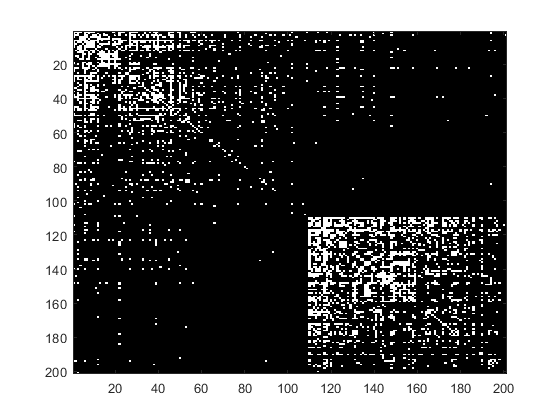}\includegraphics[width=0.45\textwidth]{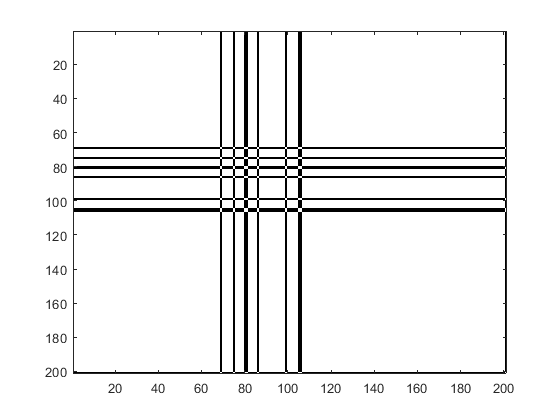}
    \caption{{\bf Left}: the adjacency matrix, in which each white entry represents an edge. {\bf Right}: comparison between the recovered $\text{sign}({\my})$ and the ground truth $\mya$, in which each black entry represents an error.}
    \label{fig:realdata}
\end{figure}

\subsection{Larger Number of Entities}
Here we provide synthetic experiment results for a large number of entities with $n=5000$ in the latent space model with two clusters. We pick $\mathcal{X} = \R^2$ as the latent domain, $\norm{\vct{\mu}}$ to be $1$, and the number of trials to be $10$. 
We compute the second minimum eigenvalue with $\sigma$ being $0.05$ and $0.3$. With $\sigma = 0.05$, the number of runs with positive second minimum eigenvalue is $10$ (out of $10$). With $\sigma = 0.3$, the number of runs with positive second minimum eigenvalue is $0$ (out of $10$). We also run SDP for both cases. With $\sigma = 0.05$ the number of runs where SDP succeeded is $10$ (out of $10$). With $\sigma = 0.3$ the number of runs where SDP succeeded is $0$ (out of $10$). Both results (success for $\sigma=0.05$ and failure for $\sigma=0.3$) confirm our finding in Theorem \ref{sdp:theorem}.


\subsection{Real-world Data}
To test the adequacy of SDP in a real-world dataset in which assumptions might not necessarily hold, we use an openly available Stanford large network dataset, \emph{email-Eu-core} \citep{snapnets}. In our experiments we used CVX \citep{cvx, gb08} as the solver.

The procedure is as follows. We select the two largest clusters from the dataset as the test data. The size of the test data is $n = 201$, and the sizes of the two clusters are $109$ and $92$, respectively. The adjacency matrix is shown in Figure \ref{fig:realdata}. Note that in the diagonal blocks in the adjacency matrix, the distribution of edges is not uniform, and seem to depend highly on the entities. That is, some rows are more dense than other rows, indicating that some entities might be closer (in a latent space) to other entities. We run SDP with the adjacency matrix and obtain the solution ${\my}$. We then set $\my = \text{sign}({\my})$ as the output of the algorithm.
Comparing our test result with the ground truth, our algorithm achieved an accuracy of 95.52\%.

For comparison, we ran the same real-world experiment using Kernighan-Lin algorithm with random initialization for $100$ iterations. The average accuracy was 52.91\%, with a standard error of 0.21\%.


\end{document}